\documentclass{birkjour}
\usepackage [latin1]{inputenc}

\newtheorem{theorem}{Theorem}[section]
\newtheorem{corollary}[theorem]{Corollary}
\newtheorem{lemma}[theorem]{Lemma}

\theoremstyle{definition}

\theoremstyle{remark}
\newtheorem{remark}[theorem]{Remark}

\numberwithin{equation}{section}

\newcommand{\BibTeX}{B\kern-0.1emi\kern-0.017emb\kern-0.15em\TeX}
\newcommand{\XYpic}{$\mathrm{X\kern-0.3em\raisebox{-0.18em}{Y}}$-$\mathrm{pic}\,$}

\newcommand{\cl}{C \kern -0.1em \ell}  

\newcommand{\ed}{\end{document}}

\begin{document}

\setcounter{page}{1} \setlength{\unitlength}{1mm}\baselineskip
.58cm \pagenumbering{arabic} \numberwithin{equation}{section}

\title[Perfect fluid spacetimes ]
{Perfect fluid spacetimes and gradient solitons}

\author[ K. De, U. C. De, A. A. Syied, N. B. Turki and S. Alsaeed ]
{ Krishnendu De$^{*}$, Uday Chand De, Abdallah Abdelhameed Syied, Nasser Bin Turki and Suliman Alsaeed }

\address
 {Department of Mathematics,
 Kabi Sukanta Mahavidyalaya,
The University of Burdwan.
Bhadreswar, P.O.-Angus, Hooghly,
Pin 712221, West Bengal, India. ORCID iD: https://orcid.org/0000-0001-6520-4520}
\email{krishnendu.de@outlook.in }

\address
{Department of Mathematics, University of Calcutta, West Bengal, India. ORCID iD: https://orcid.org/0000-0002-8990-4609}
\email {uc$_{-}$de@yahoo.com}
\address{ Department of Mathematics, Faculty of Science,
Zagazig \\
University, Egypt,}
\email{a.a\_syied@yahoo.com}

\address{Department of Mathematics, College of Science, King Saud
University, P.O. Box 2455,Riyadh 11451, Saudi Arabia, }
\email{nassert@ksu.edu.sa}
\address{Department of Mathematics, Applied Science College, Umm Al-Qura University, P.O.
Box 715, 21955 Makkah,, Saudi Arabia, }
\email{sasaeed@uqu.edu.sa}

\footnotetext {$\bf{2020\ Mathematics\ Subject\ Classification\:}.$ 53C50, 53E20, 53C35, 53E40.
\\ {Key words and phrases: Perfect fluid spacetimes, gradient Ricci solitons, gradient Yamabe solitons, $m$-quasi Einstein solitons.\\
\thanks{$^{*}$ Corresponding author}
}}
\maketitle

\vspace{1cm}

\begin{abstract}
This article deals with the investigation of perfect fluid spacetimes endowed with concircular vector field. It is shown that in a perfect fluid spacetime with concircular vector field, the velocity vector field annihilates the conformal curvature tensor and in dimension 4, a perfect fluid spacetime is a generalized Robertson-Walker spacetime with Einstein fibre. Moreover, we prove that if a perfect fluid spacetime equipped with concircular vector field admits a second order symmetric parallel tensor, then either the state equation of the perfect fluid spacetime is characterized by $p=\frac{3-n}{n-1}\sigma$ , or the tensor is a constant multiple of the metric tensor. We also characterize the perfect fluid spacetimes with concircular vector field whose Lorentzian metrics are Ricci soliton, gradient Ricci soliton, gradient Yamabe solitons and gradient $m$-quasi Einstein solitons, respectively.
\end{abstract}

\maketitle

\section{Introduction}

Let $M^{n}$ be a Lorentzian manifold endowed with the Lorentzian metric $g$ of signature $(\underbrace{+, +, \ldots, + }_{(n-1)\text{times}}, -)$. The idea of generalized Robertson-Walker $(GRW)$ spacetimes was presented by Alias, Romero and Sanchez \cite{alias1} in $1995$. A Lorentzian manifold $M^{n}$ with $n \ge 3$ is named as a $GRW$ spacetime if it can be written as a warped product of an open interval $I$ of $\mathbb{R}$ (set of real numbers) and a Riemannian manifold $ \mathcal{M^{*}}$ of dimension $(n-1)$, that is, $M=-I \times  {\mathfrak{f}^2}  \mathcal{M^{*}}$, where $\mathfrak{f}>0$ is a smooth function, termed as scale factor or warping function. If the dimension of $\mathcal{M^{*}}$ is three and is of constant sectional curvature, then the spacetime reduces to Robertson-Walker $(RW)$ spacetime. Hence, the $GRW$ spacetime is a spontaneous extension of $RW$ spacetime on which the standard cosmology is modeled. It also includes the Einstein-de Sitter spacetime, the static Einstein spacetime, the Friedman cosmological models, the de Sitter spacetime, and have implementations as inhomogeneous spacetimes obeying an isotropic radiation. The geometrical and physical features of $GRW$ spacetimes have been exhaustively presented in ( \cite{bychen}, \cite{survey}).\par

A Lorentzian manifold $M^n$ is called a perfect fluid spacetime if its non-vanishing Ricci tensor $S$ obeys
\begin{equation}
\label{1}
S=\alpha g+\beta A \otimes A,
\end{equation}
where $\alpha$, $\beta$ are scalar fields (not simultaneously zero), $\rho$ is a vector field defined by $g(X, \rho)=A(X)$ for all $X$. Also, $\rho$ is the unit timelike vector field (also named velocity vector field) of the perfect fluid spacetime. Every $RW$ spacetime is a perfect fluid spacetime \cite{neil}, where as in $4-$ dimension, the $GRW$ spacetime is a perfect fluid spacetime if and only if it is a $RW$ spacetime. In differential geometry, the Ricci tensor obeying equation (\ref{1}) is termed as a quasi Einstein manifold \cite{chaki}. For more details, we refer (\cite{blaga2}, \cite{sharma1}) and the references therein.\par

The problem of discovering a canonical metric on a smooth manifold inspires Hamilton \cite{rsh2} to introduce the concept of Ricci flow. If the metric of a {(semi-) Riemannian} manifold $M^n$ is satisfied by an evolution equation $\frac{\partial}{\partial t}g_{ij}(t)=-2S_{ij}$, then it is called Ricci flow \cite{rsh2}. The self-similar solutions to the Ricci flow yield the Ricci solitons. A metric of $M^n$ is called a Ricci soliton \cite{rsh1} if it obeys
\begin{equation}
\label{4}
\mathfrak{L}_{W}g+2S+2\lambda g=0
\end{equation}
for some real scalar $\lambda$. Here $\mathfrak{L}_{W}$ indicates the Lie derivative operator. We indicate $(g, W, \lambda)$ as a Ricci soliton on $M^n$. If $\lambda$ is negative, positive or zero, then the Ricci soliton is said to be shrinking, expanding or steady, respectively. In particular, if $W$ is Killing or identically zero, then the Ricci soliton is trivial and $M^n$ is Einstein.
Also, if the soliton vector $W$ is the gradient of some smooth function $-f$, that is, $W=-Df$, then equation (\ref{4}) takes the form
\begin{equation}
\label{int1.2}
Hess \, f-S-\lambda g=0,
\end{equation}
where $Hess$ and $D$ indicates the Hessian and the gradient operator of $g$ respectively. The metric obeying equation (\ref{int1.2}) is called a gradient Ricci soliton. The smooth function $-f$ is said to be the potential function of the gradient Ricci soliton.\par

Inspired by the Yamabe's conjecture (``metric of a complete Riemannian manifold is conformally related to a metric with constant scalar curvature"), the idea of Yamabe flow on a complete Riemannian manifold $M^n$ was presented by Hamilton \cite{rsh2}. A semi-Riemannian manifold $M^n$ endowed with a semi-Riemannian metric $g$ is called a Yamabe flow if it obeys:
\begin{equation*}
\frac{\partial}{\partial t}g(t)=-r g(t), \,\,\,\,\, g_{0}=g(t),
\end{equation*}
where $t$ indicates the time and $r$ is the scalar curvature of $M$. A semi-Riemannian manifold $M^n$ equipped with a semi-Riemannian metric $g$ is called a Yamabe soliton if it obeys
\begin{equation}
\label{yb1}
\frac{1}{2} \mathcal{L}_{W}g=(r-\lambda)g
\end{equation}
for real constant $\lambda : M \rightarrow \mathbb{R}$. Here $\mathcal{L}$ indicates the Lie derivative operator, $W$ is a vector field, termed as the potential vector field and $\mathbb{R}$ is the set of real numbers. Yamabe soliton with $W=Df$ reduces to the gradient Yamabe soliton on semi-Riemannian manifold $M^n$. Thus, equation (\ref{yb1}) takes the form
\begin{equation}
\label{yb2}
 Hess f=(r-\lambda)g.
\end{equation}
If $f$ is constant (or $W$ is Killing) on $M$, then gradient Yamabe (or Yamabe) soliton becomes trivial. Sharma \cite{sharma1} investigated the Yamabe soliton on $3$-Sasakian manifolds. Also, the $3$-Kenmotsu
manifolds and almost co-K\"{a}hler manifolds with Yamabe solitons have been characterized  by Wang \cite{yw1} and Suh and De \cite{suh2} respectively. Chen et al. \cite{chen1} studied the properties of Riemannian manifolds with Yamabe solitons. Some interesting results on Yamabe solitons have been investigated in (\cite{blaga2}, \cite{kde1}, \cite{kde2}) and also by others. Recently, in \cite{ude}, the authors have studied Yamabe and gradient Yamabe solitons in perfect fluid spacetimes.\par

A semi-Riemannian manifold $M^n$ equipped with the semi-Riemannian metric $g$ is said to be a \emph{gradient m-quasi Einstein metric} \cite{br1} if there exists a constant $\lambda$, a smooth function $f:M^{n}\rightarrow \mathbb{R}$ and obeys
\begin{equation} S+Hess f=\frac{1}{m}df \otimes df +\lambda g,\label{a4}\end{equation}
where $0 < m\leq \infty $ is an integer and $\otimes$ indicate the tensor product. In this case $f$ denotes the $m$-quasi Einstein potential function \cite{br1}. Here the Bakry-Emery Ricci tensor $S+Hess f-\frac{1}{m}df \otimes df$ is proportional to the metric $g$ and $\lambda = constant$ \cite{ww}.

If $m = \infty$, the foregoing equation (\ref{a4}) represents a gradient Ricci soliton and the metric represents almost gradient Ricci soliton if it obeys the condition $m = \infty$ and $\lambda$ is a smooth function. Some basic classifications of $m$-quasi Einstein metrics was characterized by He et al. \cite{he}  on Einstein product manifold with non-empty base. Also, a few characterization of $m$-quasi Einstein solitons have been presented (in details) in \cite{hu}.\par

The above studies motivate us to study the properties of perfect fluid spacetimes if the Lorentzian metrics are Ricci, gradient Ricci, gradient Yamabe and $m$-quasi Einstein solitons. We lay out the content of our paper as:\par

In Section $2$, we produce the preliminaries idea of perfect fluid spacetime with concircular vector field. The properties of second order symmetric parallel tensor in perfect fluid spacetimes with concircular vector field are studied in Section $3$. Section $4$ concerns with Ricci soliton and gradient Ricci soliton on a perfect fluid spacetime with concircular vector field. We investigate the properties of perfect fluid  spacetimes equipped with gradient yamabe soliton and gradient $m$-quasi Einstein solitons in Section $5$ and Section $6$, respectively.

\section{Perfect fluid  spacetime}
 It is well known that in a perfect fluid spacetimes, $\rho$ is the unit timelike vector field (also termed as velocity  vector field ), hence
\begin{equation}
\label{2}
g(U, \rho)=A(U),\,\,\,\,\,\,\, g(\rho, \rho)=A(\rho)=-1,
\end{equation}
where the vector field $U \in \mathfrak{X}(M)$ ($\mathfrak{X}(M)$ indicates the collection of all $C^{\infty}$ vector fields of $M$) and $A$ indicates the $1$-form. Executing the covariant derivative of (\ref{2}) yields
\begin{equation}
\label{2.1a}
g(\nabla_{U}\rho, \rho)=0\,\,\, \text{and}\,\,\, (\nabla_{U}A)(\rho)=0,
\end{equation}
where $\nabla$ is the Levi-Civita connection.
The Einstein's field equations without cosmological constant have the form
\begin{equation}
\label{1.2}
S-\frac{r}{2}g=\kappa T,
\end{equation}
where $\kappa$ and $T$ indicates the gravitational constant and the energy momentum tensor, respectively. In case of perfect fluid spacetime, the \emph{energy momentum tensor} $T$ is defined as
\begin{equation}
\label{1.1}
T=(\sigma+p)A \otimes A+p g,
\end{equation}
where $\sigma$ indicates the energy density of the perfect fluid and $p$ is the isotropic pressure.

In a perfect fluid spacetime if we consider an orthonormal frame field and taking contraction of the equation (\ref{1}) over $U$ and $V$, we obtain
\begin{equation}
\label{1.4}
r=n \alpha-\beta.
\end{equation}
The necessary and sufficient condition for the constant scalar curvature of a perfect fluid spacetime is that $n U(\alpha)=U(\beta).$
Combining the equations (\ref{1}), (\ref{1.2}) and (\ref{1.1}), we infer that
\begin{equation}
\label{1.7}
\beta=\kappa(p+\sigma), \,\, \alpha=\frac{\kappa(p-\sigma)}{2-n}.
\end{equation}

Moreover, $p$ and $\sigma$ are interconnected by an equation of state of the form $p = p(\sigma )$ characterizing the particular sort of perfect fluid under consideration. In this instance, the perfect fluid is called isentropic. In addition, if $p = \sigma$, then the perfect fluid is named as stiff matter. Many years ago, a stiff matter equation of state was publicized by Zeldovich \cite{ze1}. The stiff matter era preceded the dust matter era with $p = 0$ , the radiation and it characterizes the early universe with $p -\frac{\sigma}{3}=0$ and the dark energy era with $p+\sigma =0$ \cite{ch1}.\par

The idea of concircular vector field was introduced by Failkow \cite{fi}. On a semi-Riemannian manifold $M$, a vector field $\rho$ is called concircular if there exists a smooth function $\mu$ (termed as potential function of the concircular vector field) such that
\begin{equation*}
\label{con1}
\nabla_{U}\rho=\mu U, \,\, \forall \,\, U\in \mathfrak{X}(M).
\end{equation*}
In \cite{ht}, we see that the world lines of receding or colliding galaxies in de Sitter's model of general relativity are trajectories of timelike concircular vector fields. Here, $\rho$ is called non-trivial if $\rho$ is non-constant. The vector field $\rho$ becomes concurrent vector field if $\mu$ is non-zero constant. In \cite{chen2015}, B. Y. Chen has investigated concircular vector fields and their applications to Ricci solitons. Also, Deshmukh et al.\cite{desh} have studied spheres and Euclidean spaces with the Concircular vector fields. For more information, see (\cite{des2018}, \cite{des2020}) and references contained in those.\par

Utilizing the equations (\ref{1}) and (\ref{2}), we find that
\begin{equation}
\label{11}
S(U, \rho)=(\alpha-\beta)A(U)
\end{equation}
and conclude that corresponding to the eigenvector $\rho$, $\alpha-\beta$ is an eigenvalue of the Ricci tensor.\par

{\bf Agreement:} Throughout the paper, in a perfect fluid spacetime, we consider the velocity vector field is of concircular type.\par

If the velocity vector field $\rho$ of the perfect fluid spacetime is a concircular vector field, then we have $$\nabla_{U} \rho=\mu U$$ for all $U$. Utilizing the above expression together with $R(U, V)\rho=\nabla_{U} \nabla_{V}\rho-\nabla_{V} \nabla_{U}\rho-\nabla_{[U, V]} \rho$ yield
\begin{equation}
\label{12}
R(U, V)\rho=(U \mu)V-(V\mu)U.
\end{equation}
Executing contraction over $U$, the foregoing equation gives
\begin{equation}
\label{13}
S(V, \rho)=(1-n)(V \mu).
\end{equation}
Combining the equations (\ref{11}) and (\ref{13}), we infer that
\begin{equation}
\label{14}
(U \mu)=\frac{\alpha-\beta}{1-n}A(U).
\end{equation}
Using (\ref{14}) into equation (\ref{12}), we find
\begin{equation}
\label{15}
R(U, V)\rho=\frac{\alpha-\beta}{1-n}\{A(U)V-A(V)U\}.
\end{equation}
Applying the above equation and from the subsequent expression of Weyl conformal curvature tensor \cite{hw}
\begin{eqnarray*}
\label{}
&&C(U, V)X=R(U, V)X-\frac{1}{n-2}[g(V, X)QU-g(U, X)QV+S(V, X)U\nonumber\\&&
\,\,\,\,\,\,\,\,\,\,\,\,\,\,\,\,\,\,\,\,\,\,\,\,\,\,\,\,\,\,-S(U, X)V-\frac{r}{n-1}\{g(V, X)U-g(U, X)V\}],
\end{eqnarray*}
where $R$ is the curvature tensor and Ricci operator $Q$ is defined by $g(QU, V)=S(U, V)$, we obtain that
\begin{equation*}
\label{}
C(U, V)\rho=0, \,\, \forall \,\,U, V \in \mathfrak{X}(M).
\end{equation*}
\begin{theorem}
\label{thm2.1}
In a perfect fluid spacetime with concircular vector field, the velocity vector field annihilates the conformal curvature tensor.
\end{theorem}

We know that \cite{love}, in dimension 4, $C(U,V)\rho=0$ is equivalent to $C(U,V)W=0$. Also, from $\nabla_{U} \rho=\mu U$, we get
$$(\nabla_{U} A)(V)=\mu g(U,V)=(\nabla_{V} A)(U).$$
Hence, the 1-form $A$ is closed.\par
In \cite{manticamolinaride}, Mantica et al proved that in a 4-dimensional perfect fluid spacetime satisfying $div \mathbf{C}=0$ is a $GRW$ spacetime with Einstein fibre, provided the velocity vector field $\rho$ is irrotational.\par
Since  $C(U,V)W=0$, then $div \mathbf{C}=0$. Hence, from the above discussion, we say that the perfect fluid spacetime with concircular vector field is a $GRW$ spacetime with Einstein fibre.
 Thus, we can write:
\begin{theorem}
\label{thm2.2}
A perfect fluid spacetime with concircular vector field is a $GRW$ spacetime with Einstein fibre.
\end{theorem}

\section{Second order symmetric parallel tensor in a perfect fluid spacetime}
Let $P $ be a symmetric (0,2)-tensor in a perfect fluid spacetime which is parallel with respect to the Levi-civita connection $\nabla$, that is $\nabla P =0$.
Then, by $\nabla P =0$, we get
\begin{equation} P (R(U,V)X,Y)+P(X,R(U,V)Y)=0,\label{c1}\end{equation}
where $U,V,X$ and $Y$ are arbitrary vectors fields.
As $P$ is symmetric , putting $=X=Y=\rho $ in
(\ref{c1}), we have
\begin{equation} P(R(U ,V)\rho, \rho )=0.\label{c2}\end{equation}
Utilizing (\ref{15}), we obtain
\begin{equation}
\label{c3}
\frac{\alpha-\beta}{1-n} P(A(U)V-A(V)U,\rho)=0.
\end{equation}
Replacing $U$ by $\rho$ in the foregoing equation and using (\ref{2}), we infer
\begin{equation}
\label{c4}
\frac{\alpha-\beta}{1-n} \{-P(V,\rho)-A(V)P(\rho,\rho)\}=0,
\end{equation}
which entails that either $\alpha=\beta$, or
\begin{equation}
\label{c5}
P(V,\rho)=-A(V)P(\rho,\rho).
\end{equation}
Since $P$ is parallel, we have
\begin{eqnarray}
0=(\nabla_{X}P)(V,\rho)&=&\nabla_{X} P(V,\rho)-P(\nabla_{X}V,\rho)-P(V,\nabla_{X}\rho)\nonumber\\&&
=-\nabla_{X} A(V)P(\rho,\rho)+A(\nabla_{X}V)P(\rho,\rho)-P(V,\mu X)\nonumber\\&&
=-\nabla_{X} A(V)P(\rho,\rho)-\mu P(V,X)\nonumber\\&&
=-\mu g(X,V)P(\rho,\rho)-\mu P(V,X)\nonumber.\end{eqnarray}
Since $\mu \neq 0$, $$g(X,V)P(\rho,\rho)=P(V,X),$$
which implies that
$$g(X,V)(\nabla_{Y}P)(\rho,\rho)=(\nabla_{Y}P)(V,X).$$
Since $\nabla P=0$, we conclude that $P(\rho,\rho)=$ constant.\par

From (\ref{15}) we derive
\begin{equation}
\label{c6}
R(U,\rho)V=\frac{\alpha-\beta}{1-n}\{g(U,V)\rho-A(V)U\}.
\end{equation}
Putting $V=Y=\rho$ in (\ref{c1}) and using (\ref{c6}), we lead either $\alpha=\beta$, or $$P(X,U)=P(\rho,\rho)g(X,U).$$
When $\alpha=\beta$,  from equation (\ref{1.7}), we find
\begin{equation*}
\label{}
p=\frac{3-n}{n-1}\sigma,
\end{equation*}
which gives the equation of state  in a perfect fluid spacetime.
Therefore, we can write:
\begin{theorem}
\label{thmc3.1}
 If a perfect fluid spacetimes admits a second order symmetric parallel tensor, then either the state equation of the perfect fluid spacetime is characterized by $p=\frac{3-n}{n-1}\sigma$ , or the tensor $P$ is a constant multiple of the metric tensor.
\end{theorem}
\begin{remark}
 For $n =4$, we get the state equation as $\sigma+3p=0$, which implies the radiation and it characterizes the early universe.
\end{remark}
\begin{corollary}
\label{cor2.2}
If a perfect fluid spacetime is Ricci symmetric ($\nabla S=0$), then the manifold is an Einstein manifold, provided $\alpha \neq \beta$.
\end{corollary}

\section{  Ricci soliton and gradient Ricci soliton on a perfect fluid spacetime}

$\;\;\;\;$ Suppose that a perfect fluid spacetime with concircular vector field admits a Ricci soliton defined by $(\ref{4})$. We know that $\nabla g=0$ and $\nabla \lambda g=0,$ since $\lambda $ in the equation $(\ref{4})$ is a constant. Therefore, $\pounds _{W}g+2S$ is parallel. Hence utilizing the preceding theorem we obtain $\pounds _{W}g+2S $ is a constant multiple of metric tensors $g$, that is, $\pounds _{W}g+2S=a_{1} g,$ where $a_{1}$ is constant, provided $\alpha \neq \beta$. Hence $\pounds _{V}g+2S+2\lambda g$ takes the form $(a_{1}+2\lambda )g$, that implies $\lambda =-a_{1}/2$. Thus we
write the subsequent result:
\begin{theorem}
In a perfect fluid spacetime equipped with concircular vector field, the Ricci soliton  $(g,V,\lambda )$ is expanding or shrinking according as $a_{1}$ is negative or positive, provided $\alpha \neq \beta$.\end{theorem}

Now, particularly we study the case $W=\rho .$ Then $(\ref{4})$ takes the form
\begin{equation} (\pounds_{\rho }g)(U,V)+2S(U,V)+2\lambda g(U,V)=0.\label{f1}\end{equation}
Utilizing $\nabla_{U} \rho=\mu U$, we infer
\begin{equation} S(U,V)=-(\mu+\lambda)g(U,V),\label{f2}\end{equation}
which implies a Einstein spacetime, that is, trivial Ricci soliton.
Therefore, we have
\begin{theorem} In a perfect fluid spacetime endowed with concircular vector field, the Ricci soliton $(g,\rho ,\lambda )$ is trivial and is an Einstein spacetime, provided $\alpha \neq \beta$.\end{theorem}

Next part of this section deals with the investigation of gradient Ricci solitons in perfect fluid spacetimes with concircular vector field. \par

Now,
\begin{eqnarray}(\nabla_{V}A)(U)&=&\nabla_{V} A(U)-A(\nabla_{V}U)\nonumber\\&&
=\nabla_{V} g(U, \rho)-g(\nabla_{V}U, \rho)\nonumber\\&&
=g(U, \nabla_{V}\rho)=\mu g(U,V)\label{3.1},\end{eqnarray}
$\forall \,\, U, \,V \in \mathfrak{X}(M)$ and from (\ref{1}), we obtain

\begin{equation}
\label{3.2}
QU=\alpha U+\beta A(U) \rho,\,\, \forall\,\, U\in \mathfrak{X}(M).
\end{equation}
Let us assume that the soliton vector field $W$ of the Ricci soliton $(g, W, \lambda)$ in a perfect fluid spacetime with concircular vector field is a gradient of some smooth function $-f$. Then equation (\ref{4}) reduces to
\begin{equation}
\label{3.3}
\nabla_{U}Df=QU+\lambda U
\end{equation}
for all $U \in \mathfrak{X}(M)$.
The equation (\ref{3.3}) along with the subsequent relation
\begin{equation}
\label{3.4}
R(U, V)Df=\nabla_{U} \nabla_{V}Df-\nabla_{V} \nabla_{U}Df-\nabla_{[U, V]}Df
\end{equation}
yield
\begin{equation}
\label{3.5}
R(U, V)Df=(\nabla_{U}Q)(V)-(\nabla_{V}Q)(U).
\end{equation}
Executing the covariant derivative of (\ref{3.2}) and utilizing (\ref{3.1}), we lead
\begin{equation}
\label{3.6}
(\nabla_{U}Q)(V)=(U \alpha)V+(U \beta) A(V)\rho+\beta (\mu g(U,V)\rho+\mu A(V)U).
\end{equation}
In view of equations (\ref{3.5}) and (\ref{3.6}), we infer
\begin{eqnarray}
\label{3.7}
&&R(U, V)Df=(U \alpha)V-(V \alpha)U+\{(U \beta)A(V)-(V \beta)A(U)\}\rho\nonumber\\&&
\,\,\,\,\,\,\,\,\,\,\,\,\,\,\,\,+\mu \beta \{A(V)U-A(U)V\}.
\end{eqnarray}
Taking a set of orthonormal frame field and executing contraction of the equation (\ref{3.7}), we get
\begin{equation}
\label{3.8}
S(U, Df)=(1-n)(U \alpha)+(U \beta)+(\rho \beta)A(U)+\mu \beta(n-1)A(U).
\end{equation}
Again, from (\ref{1}) we have
\begin{equation}
\label{3.9}
S(V, Df)=\alpha (V f)+\beta A(V) (\rho f).
\end{equation}
Setting $V=\rho$ in equations (\ref{3.8}) and (\ref{3.9}) and then equating the values of $S(\rho, Df)$, we find
\begin{equation}
\label{3.10}
(\alpha-\beta) (\rho f)=(1-n) \{(\rho \alpha)-\mu \beta\}.
\end{equation}
Suppose that $$(\rho \alpha)=\mu \beta.$$ Thus, from equations (\ref{3.10}) we get
\begin{equation}
\label{3.11}
(\alpha-\beta) (\rho f)=0,
\end{equation}
which shows that either $\alpha=\beta$ or $(\rho f)=0$ on a perfect fluid spacetime with the gradient Ricci soliton.\par

{\bf Case I}.
We assume that $\alpha=\beta$ and ($\rho f) \ne 0$ and therefore from equation (\ref{1.7}), we conclude that
\begin{equation*}
\label{}
p=\frac{3-n}{n-1}\sigma.
\end{equation*}
This gives the equation of state  in a perfect fluid spacetime. Also, $\lambda=\beta-\alpha=0$ and hence the gradient Ricci soliton is steady.\par

{\bf Case II}. We consider that ($\rho f)=0$ and $\alpha \ne \beta$.
Then, we have, $f$ is invariant under the velocity vector field $\rho$.\par

By concluding the above facts, we can write our result as:
\begin{theorem}
\label{thm3.1}
Let the perfect fluid spacetimes with concircular vector field admit a gradient Ricci soliton with $(\rho \alpha)= \mu \beta$. Then  either the state equation of the perfect fluid spacetime is governed by $p=\frac{3-n}{n-1}\sigma$ and the soliton is steady, or $f$ is invariant under the velocity vector field $\rho$.
\end{theorem}

\section{Gradient Yamabe soliton on perfect fluid spacetimes}

From equation (\ref{yb2}), we find
\begin{equation}
\label{4.1}
\nabla_{V}Df=(r-\lambda)V.
\end{equation}
Differentiating (\ref{4.1}) covariantly along the vector field $V$, we have
\begin{equation}
\label{4.2}
\nabla_{U}\nabla_{V}Df=(Ur)V+(r-\lambda)\nabla_{U}V.
\end{equation}
Interchanging $U$ and $V$ in the foregoing equation and then utilizing the above equation, (\ref{4.1}) and (\ref{4.2}) in $R(U, V)Df=\nabla_{U} \nabla_{V}Df-\nabla_{V}\nabla_{U}Df-\nabla_{[U, V]}Df$, we lead
\begin{equation*}
\label{4.3}
R(U, V)Df=(Ur)V-(Vr)U.
\end{equation*}
Considering an orthonormal frame field and contracting the above equation over $U$, we get
\begin{equation*}
\label{4.4}
S(V, Df)=-(n-1)(Vr).
\end{equation*}
From equation (\ref{1}) we have
\begin{equation*}
\label{4.5}
S(V, Df)=\alpha (Vf)+\beta (\rho f) A(V).
\end{equation*}
Combining the last two equations, we infer
\begin{equation}
\label{4.5a}
\alpha (V f)+\beta (\rho f) A(V)=-(n-1)(V r).
\end{equation}
Putting $V=\rho$ in the preceding equation, we get
\begin{equation}
\label{4.6}
(\alpha-\beta)(\rho f)=-(n-1)(\rho r).
\end{equation}

Now, from (\ref{4.3}) we infer that
\begin{equation}\label{4.7}
g(R(U,V)Df,\rho)=(U r)A(V)-(V r)A(U).
\end{equation}
Again (\ref{15}) implies that
\begin{equation}\label{4.8}
g(R(U, V)\rho,Df)=\frac{\alpha-\beta}{1-n}\{A(U)(V f)-A(V)(U f)\}.
\end{equation}
Combining equation (\ref{4.7}) and (\ref{4.8}), we have
\begin{equation}\label{4.9}
(U r)A(V)-(V r)A(U)=-\frac{\alpha-\beta}{1-n}\{A(U)(V f)-A(V)(U f)\}.
\end{equation}
Setting $V=\rho$ in the previous equation gives
\begin{equation}\label{4.10}
(U r)=\frac{\alpha-\beta}{1-n}(U f).
\end{equation}
Utilizing (\ref{4.10}) in (\ref{4.5a}) we infer that
\begin{equation}\label{4.11}
\beta \{ (U f)+(\rho f)A(U)\}=0,
\end{equation}

which entails that either $\beta =0$ or $(U f)+(\rho f)A(U)=0$.\par

 If $\beta =0$, then we infer that $\sigma +p=0$. This represents a dark energy.\par

Next, we suppose that $\beta \ne 0$ and $(U f)+(\rho f)A(U)=0$  which implies $Df=-(\rho f)\rho$.
Hence, we conclude the result as:
\begin{theorem}
\label{thm4.1}
If the Lorentzian metric of a perfect fluid spacetime equipped with concircular vector field be a gradient Yamabe soliton, then either the spacetime represents a dark energy, or the gradient of Yamabe soliton potential function is pointwise collinear with the velocity vector field of the perfect fluid spacetime.
\end{theorem}
Next, we consider that $\beta \ne 0$ on a perfect fluid spacetime with concircular vector field admitting a gradient Yamabe soliton.

The covariant derivative of $Df=-(\rho f)\rho$ yields
\begin{equation*}
\label{}
-(r-\lambda)=-\rho (\rho f)+\mu (\rho f),
\end{equation*}
where equation (\ref{4.1}) is used. If $f$ is invariant under the velocity vector field $\rho$, then we find
\begin{equation}
\label{4.12}
\lambda=r.
\end{equation}
 From the above we conclude that the scalar curvature of the manifold is constant. Hence from equation (\ref{4.10}) we have either $\alpha=\beta$ or $Df=0$.\par

{\bf Case I}.
We suppose that $\alpha=\beta$ and $ (D f) \ne 0$ and therefore from equation (\ref{1.7}), we conclude that
\begin{equation*}
\label{}
p=\frac{3-n}{n-1}\sigma.
\end{equation*}
This gives the equation of state in a perfect fluid spacetime.\par

{\bf Case II}. We consider that $(D f)=0$ and $\alpha \ne \beta$.
The equation $Df=0$ states that $f$ is constant and hence the gradient Yamabe soliton is trivial.
Thus, we can write:
\begin{corollary}
\label{cor4.1}
Let the Lorentzian metric of a perfect fluid spacetime endowed with concircular vector field admits a gradient Yamabe soliton with $\beta \ne 0$. If $f$ is invariant under the velocity vector field $\rho$, then either the state equation of the perfect fluid spacetime is governed by $p=\frac{3-n}{n-1}\sigma$
 or the gradient Yamabe soliton is trivial.
\end{corollary}
Next, If $f$ is invariant under the velocity vector field $\rho$, then equation (\ref{4.12}) is satisfied. Thus, we conclude that the nature of the flow vary according to the scalar curvature. Thus we can write the subsequent corollaries.
\begin{corollary}
\label{cor4.2}
Let a perfect fluid spacetime with concircular vector field admits a gradient Yamabe soliton with $\beta \ne 0$. If $f$ is invariant under the velocity vector field $\rho$, then the gradient Yamabe soliton is expanding, shrinking or steady according as the scalar curvature is positive, negative or zero, respectively.
\end{corollary}

\begin{corollary}
\label{cor4.3}
 If the metric of a perfect fluid spacetime equipped with concircular vector field admits a gradient Yamabe soliton with $\beta \ne 0$, then it possesses the constant scalar curvature, provided $f$ is invariant under the velocity vector field $\rho$.
\end{corollary}
\begin{remark}
In \cite{ude}, the authors have studied gradient Yamabe soliton in perfect fluid spacetimes. But in this current paper, we consider gradient Yamabe soliton in perfect fluid spacetime with an extra condition. Precisely, in a perfect fluid spacetime, we consider the velocity vector field is of concircular type, introduced by Failkow \cite{fi} and we obtain some interesting results which are different from the results of the paper \cite{ude}.
\end{remark}

\section{Gradient $m$-quasi Einstein solitons on perfect fluid spacetimes}

Here, we investigate the perfect fluid spacetimes with concircular vector field with $m$-quasi Einstein metric and at first, we prove the following result
\begin{lemma}\label{lem1}
Every perfect fluid spacetimes with concircular vector field satisfies the following:
\begin{eqnarray}\label{k2}
R(U,V)D f &=& (\nabla_{V}Q)U-(\nabla_{U} Q)V+\frac{\lambda}{m}\{ (V f)U-(U f)V\}\nonumber\\&&
+\frac{1}{m}\{ (U f)Q V-(V f)Q U \},
\end{eqnarray}
for all $U, \, V \in \mathfrak{X}(M)$.
\end{lemma}
\begin{proof}
Let us assume that the perfect fluid spacetimes (with concircular vector field) with $m$-quasi Einstein metric. Then the equation (\ref{a4}) may be expressed as
\begin{equation}\label{k3}
\nabla_{U}D f+Q U=\frac{1}{m}g(U,D f)D f+\lambda U.
\end{equation}
After executing covariant derivative of (\ref{k3}) along $V$, we get
\begin{eqnarray}\label{k4}
\nabla_{V}\nabla_{U}D f &=& -\nabla_{V}Q U+ \frac{1}{m} \nabla_{V}g(U,D f)D f\nonumber\\&& +\frac{1}{m} g(U,D f)\nabla_{V}D f+ \lambda \nabla_{V}U.
\end{eqnarray}
Exchanging $U$ and $V$ in (\ref{k4}), we lead
\begin{eqnarray}\label{k5}
\nabla_{U}\nabla_{V}D f &=& -\nabla_{U}Q  V+ \frac{1}{m}\nabla_{U}g(V,D f)D f\nonumber\\&& +\frac{1}{m} g(V,D f)\nabla_{U}D f +\lambda \nabla_{U}V
\end{eqnarray}
and
\begin{equation} \label{k6}
\nabla_{[U,V]}D f = -Q[U,V]+ \frac{1}{m}g([U,V],D f)D f+\lambda [U,V].
\end{equation}
Utilizing  (\ref{k3})-(\ref{k6}) and the symmetric property of Levi-Civita connection together
with $R(U,V)D f  =\nabla_{U}\nabla_{V}D f-\nabla_{V}\nabla_{U}D f-\nabla_{[U,V]}D f$, we lead
\begin{eqnarray}
R(U,V)D f  &=& (\nabla_{V}Q)U-(\nabla_{U}Q)V  +\frac{\lambda}{m}\{ (V f)U-(U f)V\}\nonumber\\&&
+\frac{1}{m} \{ (U f)Q V-(V f)Q U \}.\nonumber
\end{eqnarray}
\end{proof}
In view of the equations (\ref{3.2}), (\ref{3.6}) and the above Lemma, we infer
\begin{eqnarray}
\label{kk6}
&&R(U, V)Df=(U \alpha)V-(V \alpha)U+\{(U \beta)A(V)-(V \beta)A(U)\}\rho\nonumber\\&&
+\mu \beta \{A(V)U-A(U)V\}+\frac{\lambda}{m}\{(V f)U-(U f)V\}\nonumber\\&& +\frac{1}{m}\{\alpha (U f)V+\beta(U f)A(V)\rho-\alpha (V f)U-\beta(V f)A(U)\rho\}.
\end{eqnarray}
Taking a set of orthonormal frame field and executing contraction of the equation (\ref{kk6}), we get
\begin{eqnarray}
\label{k8}
S(U, Df)&=&(1-n)(U \alpha)+(U \beta)+(\rho \beta)A(U)\nonumber\\&&
+\mu \beta(n-1)A(U)+\frac{\lambda}{m}(n-1)(U f) \nonumber\\&&
+\frac{1}{m}\{\alpha (U f)+\beta(\rho f)A(U)-n \alpha (U f)+\beta(U f)\}.
\end{eqnarray}
Setting $V=\rho$ in equations (\ref{k8}) and (\ref{3.9}) and then equating the values of $S(\rho, Df)$, we find
\begin{equation}
\label{k9}
(\frac{m}{1-n}+\lambda-\alpha) (\rho f)=m \{(\rho \alpha)-\mu \beta\}.
\end{equation}
If $f$ and $\alpha$ are invariant under the velocity vector field $\rho$, then we find from the foregoing equation that $\beta =0$, since $m\neq 0$.
Thus, following the proof of Theorem 5.1, we conclude our result as:
\begin{theorem}
\label{thm4.1}
 Let a perfect fluid spacetime endowed with concircular vector field admits a gradient $m$-quasi Einstein soliton. If $f$ and $\alpha$ are invariant under the velocity vector field $\rho$, then the spacetime represents a dark energy.
\end{theorem}
\section{Conclusion}
In true sense, solitons are nothing but the waves which is physically propagate with some loss of energy and hold their speed and shape after colliding with one more such wave. In nonlinear partial differential equations describing wave propagation, solitons play an important role in the treatment of initial-value problems.\par

In this current investigation, we establish that a perfect fluid spacetime with concircular vector field is a generalized Robertson-Walker spacetime with Einstein fibre. Moreover, we prove that if a perfect fluid spacetime equipped with concircular vector field admits a second order symmetric parallel tensor, then either the state equation of the perfect fluid spacetime is characterized by $p=\frac{3-n}{n-1}\sigma$ , or the tensor is a constant multiple of the metric tensor. Also, different metrics like Ricci soliton, gradient Ricci soliton, gradient Yamabe solitons and gradient $m$-quasi Einstein solitons are studied in the perfect fluid spacetimes with concircular vector field. Specifically, we obtain the condition for which the vector field $\rho$ is steady, expanding and shrinking and observe that the spacetime represents a dark matter era under certain restriction on the vector field $\rho$.

\section{List of Abbreviations}
1.Generalized Robertson-Walker $= (GRW)$.\par
2. Robertson-Walker $= (RW)$.
\section{Declarations}
\subsection{Funding }
Not applicable.
\subsection{Conflicts of interest/Competing interests}
The authors declare that they have no conflict of interest.
\subsection{Availability of data and material }
Not applicable.
\subsection{Code availability}
Not applicable.


\end{document}